\documentclass[a4paper,10pt]{article}
\usepackage[english]{babel}
\usepackage[utf8]{inputenc}
\usepackage{amsmath}
\usepackage{amssymb}
\usepackage[english]{babel}
\usepackage{amsthm}
\usepackage{stmaryrd}
\usepackage[justification=centering]{caption}
\usepackage{algorithm}
\usepackage[nottoc]{tocbibind}
\usepackage[inline]{asymptote}
\usepackage[utf8]{inputenc}
\usepackage{algorithm}
\usepackage{algpseudocode}
\usepackage{hyperref}
\usepackage{url}
\usepackage{authblk}
\usepackage{titling}
\usepackage{titlesec}
\usepackage[english]{babel}
\usepackage[autostyle, english = american]{csquotes}
\MakeOuterQuote{"}

\algnewcommand\algorithmicgenerate{\textbf{GENERATE}}
\algnewcommand\GENERATE[1]{\State\algorithmicgenerate\ #1}
\algnewcommand\algorithmicset{\textbf{SET}}
\algnewcommand\SET[1]{\State\algorithmicset\ #1}
\algrenewcommand\algorithmicfor{\textbf{for}} 
\algrenewcommand\algorithmicend{\textbf{end}} 
\algloop{FOR}

\DeclareMathOperator*{\argmin}{arg\,min}
\newtheorem{theorem}{Theorem}[section]
\newtheorem{definition}[theorem]{Definition}

\newtheorem{lemma}[theorem]{Lemma}
\newtheorem{remark}[theorem]{Remark}

\usepackage{mathtools}
\DeclarePairedDelimiter{\ceil}{\lceil}{\rceil}

\errorcontextlines\maxdimen

\makeatletter
\newcommand*{\algrule}[1][\algorithmicindent]{\makebox[#1][l]{\hspace*{.5em}\vrule height .75\baselineskip depth .25\baselineskip}}%

\newcount\ALG@printindent@tempcnta
\def\ALG@printindent{%
    \ifnum \theALG@nested>0
        \ifx\ALG@text\ALG@x@notext
            \addvspace{-3pt}
        \else
            \unskip
            \ALG@printindent@tempcnta=1
            \loop
                \algrule[\csname ALG@ind@\the\ALG@printindent@tempcnta\endcsname]%
                \advance \ALG@printindent@tempcnta 1
            \ifnum \ALG@printindent@tempcnta<\numexpr\theALG@nested+1\relax
            \repeat
        \fi
    \fi
    }%
\usepackage{etoolbox}
\patchcmd{\ALG@doentity}{\noindent\hskip\ALG@tlm}{\ALG@printindent}{}{\errmessage{failed to patch}}
\makeatother

\title{\textbf{Minimum Path Star Topology Algorithms for Weighted Regions and Obstacles}}
\author{
  King, Tyler\\
  \texttt{ttk22@cornell.edu}
  \and
  Soltys, Michael\\
  \texttt{michael.soltys@csuci.edu}
}

\begin{document}
\date{\today}
\maketitle
\section{Abstract}
Shortest path algorithms have played a key role in the past century, paving the way for modern day GPS systems to find optimal routes along static systems in fractions of a second. One application of these algorithms includes optimizing the total distance of power lines (specifically in star topological configurations). Due to the relevancy of discovering well-connected electrical systems in certain areas, finding a minimum path that is able to account for geological features would have far-reaching consequences in lowering the cost of electric power transmission. We initialize our research by proving the convex hull as an effective bounding mechanism for star topological minimum path algorithms. Building off this bounding, we propose novel algorithms to manage certain cases that lack existing methods (weighted regions and obstacles) by discretizing Euclidean space into squares and combining pre-existing algorithms that calculate local minimums that we believe have a possibility of being the absolute minimum. We further designate ways to evaluate iterations necessary to reach some level of accuracy. Both of these novel algorithms fulfill certain niches that past literature does not cover. 

\newpage

\section{Introduction}

Throughout this paper we deal with the drawbacks of Weiszfeld's algorithm, proposing alternative algorithms that deal with specific Euclidean spaces, including weighted regions and obstacles. It is known that algorithms such as the continuous Dijkstra paradigm \cite{cdp2,latexcompanion} and A* pathfinding \cite{monotone,Mousoulides2014} are effective at managing point-point shortest path problems, but are unable to be expanded to dealing with multiple nodes. By creating novel algorithms that can approximate minimums for star topology, we hope to have effects outside of the theoretical plane. In particular, we looked to connect several cities (viewed as nodes) to a central power grid (central node, obviated in following diagrams for simplicity). Ordinarily, Weiszfeld's algorithm could be successful in optimizing this task, however, it lacks the ability to deal with geological difficulties such as national parks (obstacles) or mountainous versus flat terrains (weighted regions where more weight could be given to more challenging terrain) \cite{WeiszfeldInfo}. 

We begin by proving the bounding of Weiszfeld's algorithm as an introductory proof, establishing the tone for following proofs. We complete a proof by contradiction that shows any point $q$ outside of the convex hull can be reflected over the line formed by the edge of the convex hull to a new point $q^\prime$ (reference figure 1). From here, $q^\prime := q$, and continuously reiterating through this pattern leads to a point inside the convex hull. Each reflection moves towards the set of points $p_1$, $p_2$,..., $p_N$, thus proving that all points outside the convex hull have some reflected point inside the convex hull that is closer to the set of points.

\begin{figure}[h]
\begin{center}
\includegraphics[scale=.3]{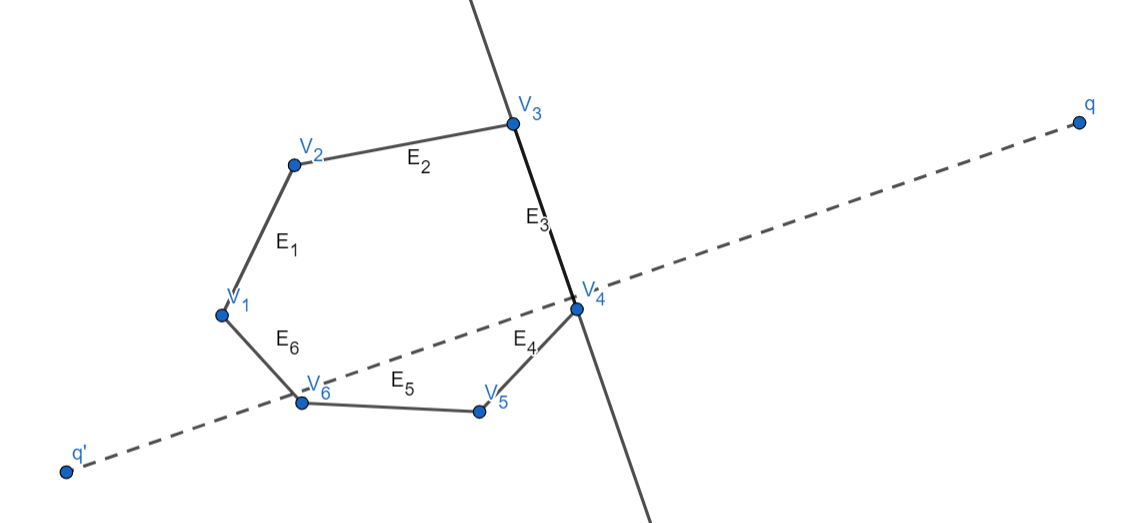}
\caption{Point $q$ mirrored to point $q^\prime$ over line $o$ which extends out from line segment $E_3$}
\end{center}
\end{figure}

For the next two proofs, We again use contradiction. We begin by parsing the area observed outside the convex hull into 4 distinct cases based on their location outside the convex hull (as seen in figure 2), and from each one of these cases, we prove that there exists some point inside or along the convex hull that is entrapping either all weighted regions+points or obstacles+points that is closer to the set of points $p_1$, $p_2$,..., $p_N$ than any point outside the convex hull.

\begin{figure}[h]
\begin{center}
\includegraphics[scale=.4]{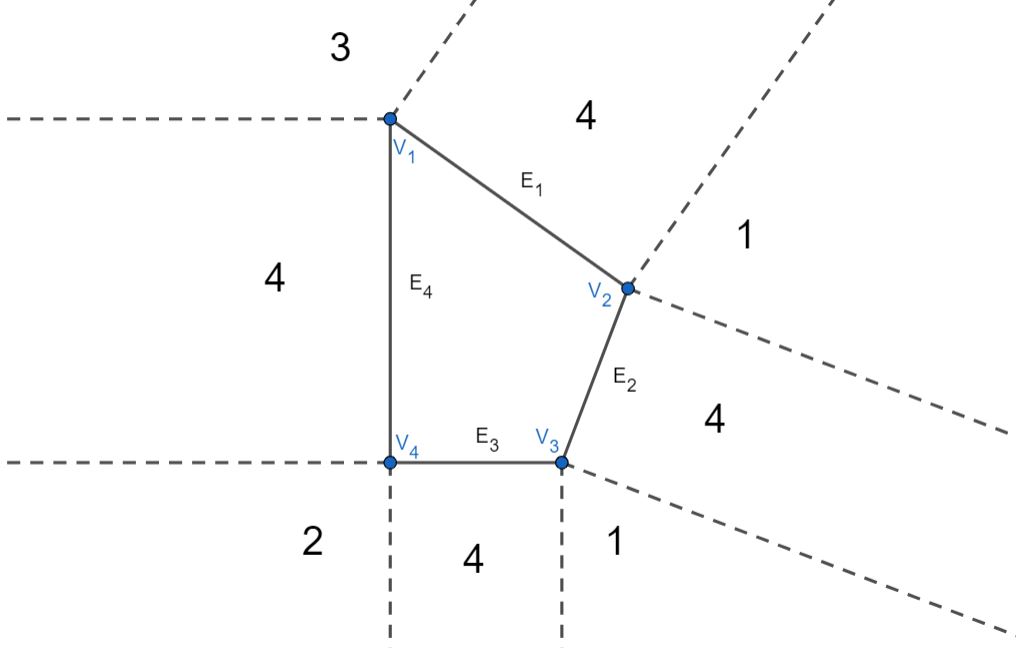}
\caption{Parsing of area outside convex hull into 4 cases}
\end{center}
\end{figure}

We propose an algorithm that deals with weighted regions. By finding a minimum cost algorithm along this surface, it becomes possible to approximate the lower bound on the pricing for a star topological configuration of power lines from a central node to a set of nodes $p_1$, $p_2$,$...$, $p_N$ by bounding the convex hull inside a rectangular object and placing a grid pattern over it. Continue onward by splitting this graph into squares with side length $n_i$ (where $i=1$ initially and increases each iteration increases accuracy), and place a node at the center of each square that takes a weighted average of the regions within the square (represented as $h_i^{(0)}$ where this $i$ counts up to the number of squares in the grid). If any of $p_1$, $p_2$,$...$, $p_N$ falls inside the square, evaluate it at the node in the center of the square. This approach is demonstrated in figure 3. 

\begin{figure}[h]
\begin{center}
\includegraphics[scale=.42]{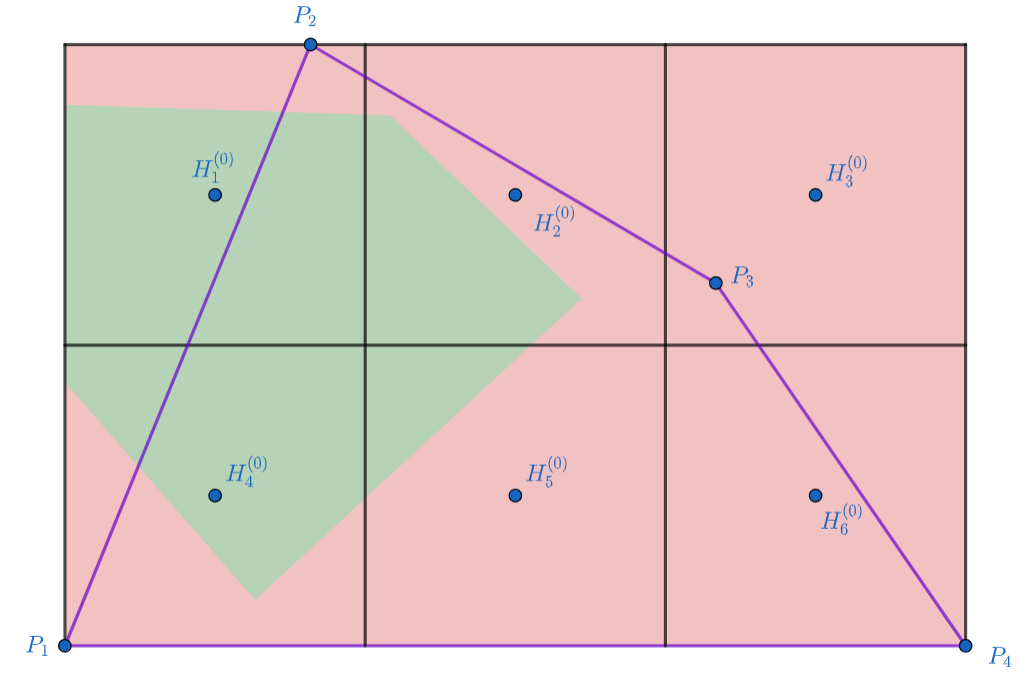}
\caption{Arbitrary region with nodes $p_1$, $p_2$, $p_3$, $p_4$ surrounded by a convex hull (in purple). Graph has two distinct weighted regions (red highlight and green highlight) and is split into 6 squares, each with its center labeled $h_i^{(0)}$. Within each square the average of the weighted regions must be calculated and any of the initial nodes $p_1$, $p_2$, $p_3$, $p_4$ are evaluated at the nearest $h_i^{(0)}$.}
\end{center}
\end{figure}

From the node that results in
$$\min \sum_{j=1}^{N} w_a(h_i^{(0)}, p_j)$$
where $w_a(x, y)$ implements a minimum monotone A* pathfinding algorithm between nodes $x$, $y$ (formalized in section 4), create a square that intersects the 8 successor nodes. From this new square, continuously reiterate through the pattern outlined earlier (splitting into new squares with side length $2n_{i+1}$), leading towards an approximation of $q$. Additionally, it is possible to incorporate Weiszfeld's algorithm once some square that intersects the 8 successor nodes contains a consistent weighting. 

With the second algorithm (which deals with obstacles), we start in a similar fashion, placing a grid over the convex hull and splitting it into squares. Certain squares were kept or removed as necessary based on whether they intersected the convex hull, and from the remaining squares, a node was placed at the center of each and labeled similarly to the first algorithm (shown in figure 4). However, from each node, the continuous Dijkstra paradigm could be applied, which dealt with Euclidean plane cases \cite{continuousMitchell}. From here, create a new square around the 8 successor nodes. Place a grid over this square and place new nodes $h_i^{(1)}$ at the center of each new square, repeating the process using the continuous Dijkstra paradigm from every new node. This algorithm runs in $O(n\log{n}\log{\frac{1}{\epsilon_c}})$ time. 

\begin{figure}[h]
\begin{center}
\includegraphics[scale=.333]{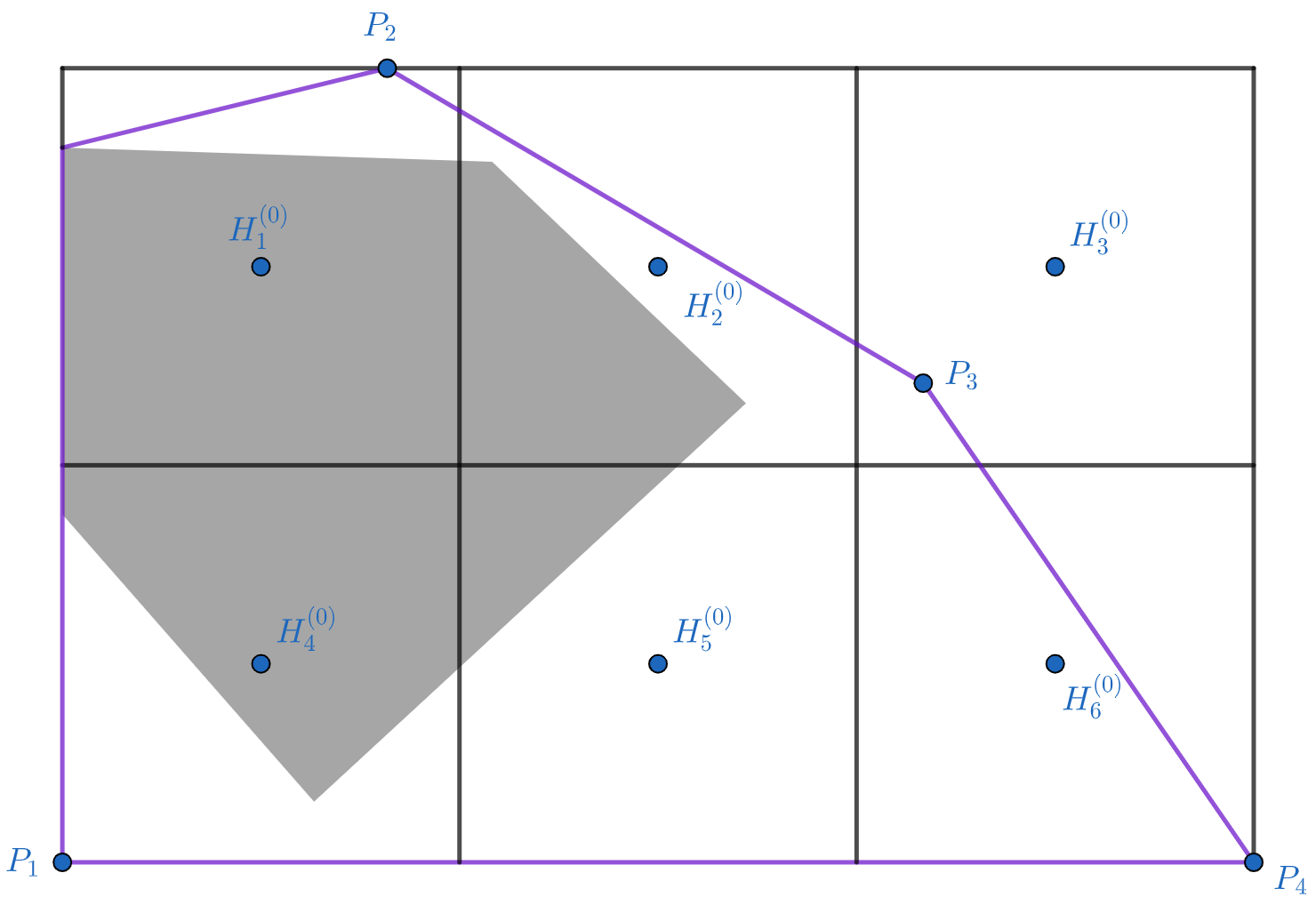}
\caption{Arbitrary region with nodes $p_1$, $p_2$, $p_3$, $p_4$ surrounded by a convex hull (in purple). Graph has one obstacle (highlighted in black) and is split into 6 squares, each with its center labeled $h_i^{(0)}$. Squares that fall within the obstacle are not evaluated.}
\end{center}
\end{figure}

For both examples, we compute an upper bound on the number of iterations needed to calculate $q$ to within some value of $\epsilon$ and produce basic psuedocode outlining both algorithms. \\
\noindent \\
\textbf{Structure of the paper.} The rest of the paper is organized as follows. Section 3 details opening and background information utilized throughout the paper. Section 4 introduces boundings for areas that must be checked to find the minimum sum of distances to several sample points. Sections 5 and 6 introduce the weighted region and obstacle problems respectively and both describe algorithms that can approximate a solution. Section 7 gives a brief overview of each algorithm. Section 8 contains conclusions. All proofs are located in the appendix, appearing in chronological order. 

\section{Literature Review}

Most content covered revolves around Wieszfeld's algorithm  \cite{einstein,Eckhardt_1980,Love2000,geometricmedian}, A* pathfinding \cite{monotone,Mousoulides2014}, continuous Dijkstra paradigm (figure 5) \cite{cdp2,latexcompanion}, and other select minimum distance pathfinding algorithms. However, there is a discontinuity in applying Weiszfeld's algorithm to weighted regions or obstacles. Although Weiszfeld's algorithm can operate in near-linear time \cite{geometricmedian} (and thus is more optimal than either of my proposed algorithms), it still cannot deal with cases that are not a consistent Euclidean plane. To begin addressing this challenge, we utilize the discretization of the Euclidean plane \cite{cdp2}, breaking it into multiple shortest path problems \cite{latexcompanion}, and then approximating the minimum by summing the independent paths. We refer to a generalized version of Weiszfeld's algorithm as star topology (or star network), where one central node connects to several other, distinct nodes. Although algorithms such as A* pathfinding and the continuous Dijkstra paradigm looked promising initially, they lacked the scalability of Weiszfeld's algorithm and thus we looked for ways to expand these algorithms to approximate a node that minimizes distance to all other nodes. An application of Weiszfeld's algorithm is demonstrated in figure 6.

\begin{figure}[h]
\begin{center}
\includegraphics[scale=.3]{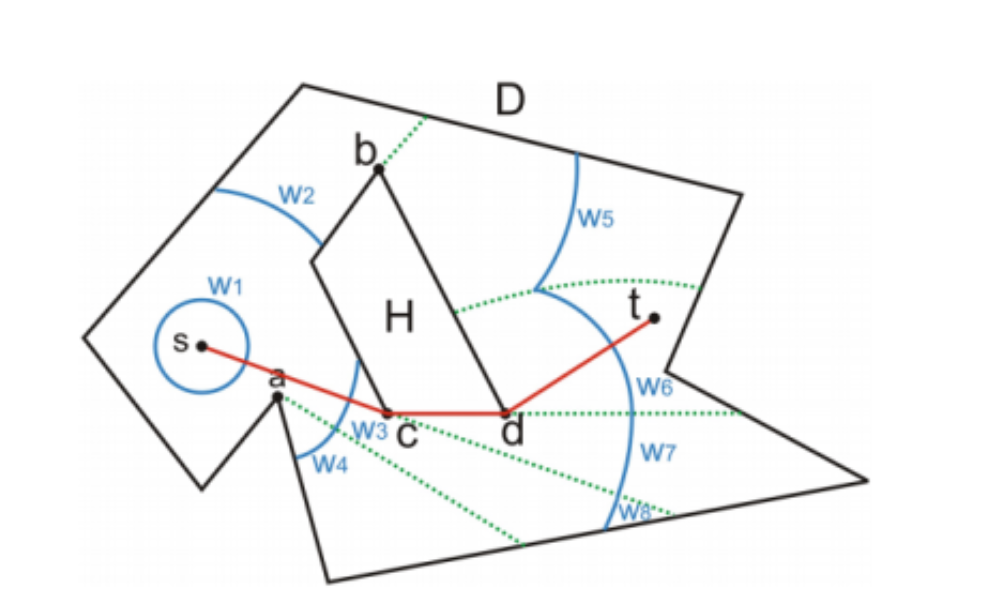}
\caption{Continuous Dijkstra paradigm finding a minimum path from node $s$ to node $t$ (each W represents another stage of the propagation of the wavelet) around obstacle $H$ \cite{latexcompanion}}
\end{center}
\end{figure}
$\\$

\begin{figure}[h]
\begin{center}
\includegraphics[scale=.24]{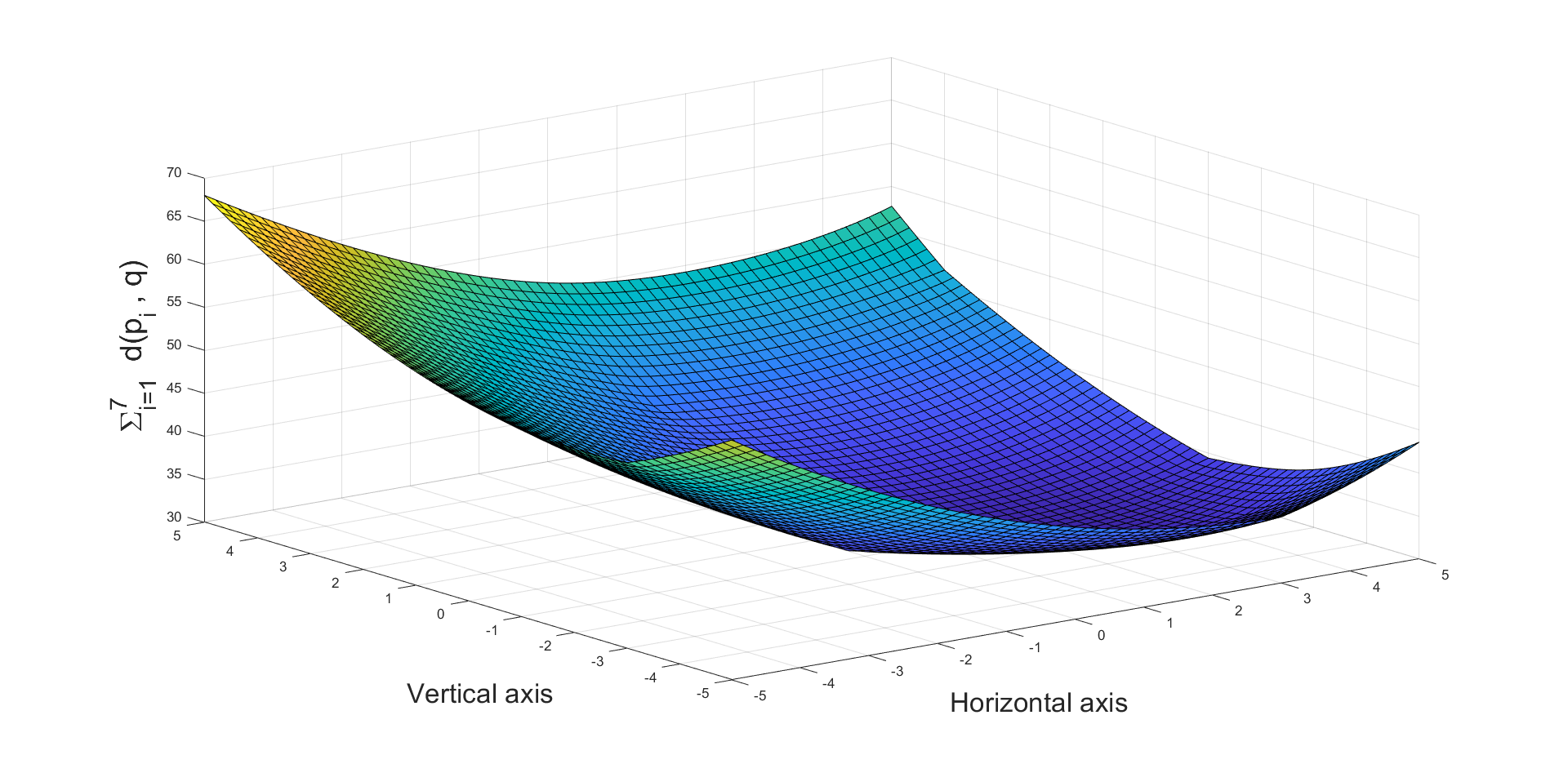}
\caption{Manifold produced by summing the distance to 7 randomly selected points. Weiszfeld's algorithm can then be applied to approximate the minimum sum along this manifold}
\end{center}
\end{figure}

\section{Bounding via Convex Hull}

\begin{definition}
$p(x, y)$ represents some arbitrary path between points $x$, $y$.
\end{definition}

\begin{definition}
$d(x, y)$ represents the numerical minimum distance between two points $x$, $y$ where $x, y \in \mathbb{R}^2$. For example, 
$$d(x, y) = \min_{\forall p(x, y)} p(x, y).$$
\end{definition}

\begin{definition}
define $\ell(A(x, y))$ as the path of the distance outlined by some function $A$ between points $x$, $y \in \mathbb{R}^2$. For example, $\ell(d(x, y))$ represents the path that is traveled along $d(x, y)$.
\end{definition}

\begin{definition}
A convex hull $C$ with finite volume and $N$ points $p_1$, $p_2$,$...$, $p_N$ can be defined as \cite{convexhull}
$$C \equiv \left\{  \sum_{j=1}^{N}\lambda_{j}\ p_j: \lambda_j \geq 0 \  \textnormal{for all}\ j\ \textnormal{and} \sum_{j=1}^{N} \lambda_j = 1\right\}.$$
\end{definition}

\begin{definition}
Assume $\sqcap$ represents the intersection of two distinct geometric constructs.
\end{definition}

\begin{lemma}
For point $q \in \mathbb{R}^2$ that falls on some side of line $o$ and points $p_1, p_2,..., p_N \in \mathbb{R}^2$ that fall on the same side of line $o$ as $q$ or on line $o$, 
$$d(p_{i}, q') \geq d(p_{i}, q)$$
where $q^\prime$ is the reflection of $q$ across $o$.
\end{lemma}

\begin{theorem}
For some set of $N+1$ points $q$, $p_1$, $p_2$,$...$, $p_N \in \mathbb{R}^2$, 
$$\min \sum_{i=1}^{N}d(q, p_i)$$ 
occurs when $q$ falls in or along $C$, where $C$ is the convex hull formed by points $p_1$, $p_2$,$...$, $p_N$.
\end{theorem}

\begin{remark}
Although this technique was effective for the scenario without obstacles or weighted regions, it did not generalize well. As a result, a new, more intensive technique was necessary to approach these challenges. However, we decided to keep these proofs both for completeness and to give insight to a technique that may have future possibilities.
\end{remark}

\begin{definition}
$$d(x, y, z) = \min_{\forall \ell(d(x, z))} d(x, z)\ \textnormal{such that}\ \ell(d(x, z)) \sqcap y \neq \emptyset\ \textnormal{for}\ x, y, z \in \mathbb{R}^2.$$ 
\end{definition}

\begin{lemma}
For some set of points $p_1$, $p_2$,$...$, $p_N$ and set of polygons whose interior can be modeled by $\mathcal{P} = \{P_1, P_2, P_3,..., P_k\}$ that all fall on or on one side of some line $o$ in the Euclidean plane, the minimum sum of distances from a new point $q$ to the set of points $p_1$, $p_2$,$...$, $p_N$ such that 
$\forall P_j \in \mathcal{P}, P_j \sqcap p_i = \emptyset and P_j \sqcap \ell(d(q, p_i)) = \emptyset$ for $i \in \llbracket 1, N \rrbracket$ must occur where $q$ is either on line $o$ or on the same side as the aforementioned set of points.
\end{lemma}

\begin{theorem}
For a set of polygons whose interiors can be modeled by $\mathcal{P} = \{P_1, P_2, P_3,..., P_k\}$ and $N+1$ points $q$, $p_1$, $p_2$,$...$, $p_N$ where 
$\forall P_j \in \mathcal{P}, P_j \sqcap p_i = \emptyset$ and $P_j \sqcap \ell(d(C_i, p_i)) \neq \emptyset$  for $i \in \llbracket 1, N \rrbracket$, 
$$\min \sum_{i=1}^{N}d(p_{i}, q)$$ 
must occur when $q$ falls in or along $C$, where $C$ is the convex hull entrapping all points $p_1$, $p_2$,$...$, $p_N$, and the set of polygons $\mathcal{P}$.
\end{theorem}

\begin{definition}
Assume $p(x, y)$ traverses some distance $D$ across $j$ polygons, where $D$ can be represented by the set of $k$ paths $\{d_1, d_2,..., d_k\}$ and $k+1$ points $\{x, y, q_1, q_2,..., q_{k-1}\}$ where $\forall d_i, d_i$ traverses across only one polygon $P_j$ and $d_i$ falls between points $q_{i-1}$, $q_i$ (which fall on the boundaries of polygons) for $i \in \llbracket 2, k-1 \rrbracket$. For $i=1$, $d_i$ falls between points $x$, $q_1$, and for $i=k$, $d_k$ falls between $q_{k-1}$, $y$. Furthermore, assume $w(i)$ represents the weighting for $d_i$ traversing across $P_j$. Then,
$$D = \sum_{i=1}^{k} d_i .$$
The minimum distance along a weighted regions surface can then be expressed by $w(x, y)$, where
$$w(x, y) = \min \sum_{i = 1}^k w(i) \cdot d_i .$$
\end{definition}

\begin{lemma}
For some finite region $\in \mathbb{R}^2$ parsed into $k$ distinct polygons $\mathcal{P} = \{P_1, P_2, P_3,..., P_k\}$ each with an independent weighting $w(k)$ and $N+1$ points $q$, $p_1$, $p_2$,$...$, $p_N$, such that $p_1$, $p_2$,$...$, $p_N$ and $\mathcal{P}$ fall on the same side or on some line $o$ in the Euclidean plane,
$$\min \sum_{i=1}^{N}w(p_i, q)$$
occurs when $q$ falls in or on $C$ where C is the space of the convex hull entrapping all points $p_1$, $p_2$,$...$, $p_N$, and the polygons $\mathcal{P}$.
\end{lemma}

\begin{theorem}
For some finite region $\in \mathbb{R}^2$ parsed into $k$ polygons $\mathcal{P} = \{P_1, P_2, P_3,..., P_k\}$ each with an independent weighting $w(k)$ and $N+1$ points $q$, $p_1$, $p_2$,$...$, $p_N$, 
$$\min \sum_{i=1}^{N}w(p_i, q)$$
occurs when $q$ falls in or along $C$, where C is the space of the convex hull entrapping all points $p_1$, $p_2$,$...$, $p_N$, and the set of polygons $\mathcal{P}.$
\end{theorem}

\section{Star Topology Weighted Regions Algorithm}
In this section, we detail our first algorithm meant to deal with weighted regions for star topology algorithms.

\begin{definition}
For some arbitrary path between nodes $n_1$, $n_N$ that traverses $N$ nodes where from each node $n_i$ it must travel to one of the 8 surrounding successor nodes $n_{i+1}$, 
$$w_a(n_1, n_N) = \min \sum_{i=1}^{N-1} d(n_i, n_{i+1}) \cdot w(i+1)\ \textnormal{and for each new i,}\ n_i := n_{i+1}$$
where each new $n_i$ has 8 new $n_{i+1}$ that can be selected, $w(i+1)$ is the weighting of the successor node $n_{i+1}$ that was chosen, and for $i \in \llbracket 1, N-1 \rrbracket, n_i \in \mathbb{R}^2$.
\end{definition}

With the knowledge that the exact solution to the aforementioned weighted regions algorithm cannot be calculated in polynomial time (except for select special cases) \cite{einstein}, an approximation algorithm is created by procuring a grid pattern on the Euclidean plane that is equally spaced by some arbitrary distance $n_1$ across the convex hull. the convex hull can be inscribed in a rectangle with side lengths $l$, $m$. The number of nodes created by this technique can be calculated as $\ceil{\frac{l}{n_1}} \cdot \ceil{\frac{m}{n_1}}$, where on the sides of the rectangle the grid extends over by 
$$\frac{\ceil{\frac{l}{n_1}} - \frac{l}{n_1}}{2}\ \textnormal{and}\ \frac{\ceil{\frac{m}{n_1}} - \frac{m}{n_1}}{2}$$
for sides $l$, $m$ respectively. For each square formed by the intersections of the grid, position a node $h_i^{(0)}$ at the center that takes the weighted average of the regions inside the square. If any $p_1$, $p_2$,$...$, $p_N$ fall within a weighted square, evaluate its position as the center node of the square. In the off chance that a point falls along an edge (or corner), randomly select the center node for one of the nearest squares. Using a monotone heuristic A* pathfinding algorithm \cite{monotone} from a central node $q$ (at each $h_i^{(0)}$) to each independent node $p_1$, $p_2$,$...$, $p_N$, a two-dimensional manifold can be constructed. Assume that the successor nodes to some node are the 8 surrounding nodes. Create a square called $s_1$ with side length $2n_1$ with the 8 surrounding nodes such that the lowest node is entrapped by this square. If an edge or corner square is selected, add squares around the square chosen such that there exists 8 surrounding squares/nodes. Each of the 8 surrounding nodes are then given a weighting based on the number of paths that run through each node. From here, parse $s_i$ (in this case $s_1$) into 
$${\left[\rule{0cm}{.75cm}
\sqrt{\ceil*{\frac{l}{n_i}} \cdot \ceil*{\frac{m}{n_i}}}
\right]\rule{0cm}{.75cm}}^2$$
(where $[\ ]$ represents the nearest integer value) new squares, each with side length
$$ n_{i+1} = \frac{2n_i}{\left[\rule{0cm}{.55cm}\sqrt{\ceil*{\frac{l}{n_i}} \cdot \ceil*{\frac{m}{n_i}}}\right]\rule{0cm}{.55cm}}\ \textnormal{for}\ i \in \mathbb{Z}^+$$
Similar to earlier, create a node at the center of each of these squares that takes a weighted average of the regions inside the square. By again using a monotone heuristic A* pathfinding algorithm from each node, the minimum cost of traversing some paths from each node can be found. Proceed to create a new square $s_{i+1}$ with side length $2n_{i+1}$ that entraps the lowest node. From here assign $l, m := 2n_{i+1}$ and increment $h_i^{(j)}$ to $h_i^{(j+1)}$, and continuously reiterate through this pattern. This results in
$$\lim_{y\to\infty}{s_y} = q$$
where $q$ is an approximation of 
$$\min \sum_{i=1}^{N}w_a(p_i, q)$$
Although higher values of $n_1$ lead to faster convergence and greater accuracy, it likewise rapidly increases complexity. Low values of $n_1$ lead to extreme granularity, resulting in a likelihood of completely missing a reasonable local minimum. \\

\subsection{Star Topology Weighted Regions Algorithm Logic}
The use of a monotone heuristic algorithm results in significantly more nodes being covered while looking for a minimum distance path. This allows for the reuse of covered nodes when attempting to expand to other nodes. A* pathfinding operates under the equation $$f(n) = g(n) + h(n)$$ where $n$ is the next node on some path,  $g(n)$ represents the cost from the start node to node $n$, and $h(n)$ is a heuristic that estimates the cost of the minimum path from $n$ to the goal. The bounding for the evaluation of this heuristic can be modeled by 
$$h(n) \leq c(n, n_p) + h(n_p)\ \textnormal{and}\ h(G) = 0$$ where $n_p$ is a successor to $n$, $c(x, y)$ is the cost of traversing from node $x$ to $y$, and $G$ is a goal node. This heuristic is defined as monotone (or consistent), and by induction, will always be admissible and therefore will never overestimate the cost of reaching the goal \cite{Mousoulides2014}.

The creation of the two-dimensional manifold plays an integral role in finding the absolute minimum of the $\ceil{\frac{l}{n_i}} \cdot \ceil{\frac{m}{n_i}}$ test points. We assume that the relatively consistent topological tendencies for the areas of interest (i.e. areas with low access to electricity) \cite{topology} necessitate that the absolute minimum on the Euclidean plane of 
$$\frac{
(2l + \ceil{\frac{l}{n_i}} - \frac{l}{n_i})(2m + \ceil{\frac{m}{n_i}} - \frac{m}{n_i})
}{4}$$
has a likelihood of occurring between the minimum node and its 8 surrounding nodes instead of between two other, non-minimal nodes. This topological consistency is the same logic as to why we use the convex hull of the points $p_1$, $p_2$,$...$, $p_N$ as a bounding for the weighted regions graph even though there is a possibility that the node $q$ that leads to
$$\min \sum_{i=1}^{N}w_a(p_i, q)$$
existing outside of the convex hull. 

The weightings for each of the 8 successor nodes that fall on $s_i$ surrounding the lowest node can be determined by parsing 
$$\sum_{i=1}^{N}w_a(p_i, q)$$
into 8 distinct values $w(1), w(2),.., w(8)$ where each $w(y)$ for $y \in \llbracket 1, 8 \rrbracket$ is correlated to one of the successor nodes and the first one of these nodes that $\ell(w_a(q, p_i))$ (which is equivalent to $\ell(w_a(p_i, q))$) hits after leaving $q$ increments the value of $w(y)$ by 1. Each  weighting $w(y)$ has a related successor node whose coordinated are stored in $c(\theta_y)$. Once the new grid with squares of side length
$$n_{i+1} = \frac{2n_i}{\left[\sqrt{\ceil*{\frac{l}{n_i}} \cdot \ceil*{\frac{m}{n_i}}}\right]}$$
is placed down, the new lowest node can be calculated by 
$$\min \sum_{y=1}^{8}w_a(\theta_y, q) \cdot w(y)$$
This pattern can be repeated over and over, converging towards a local minimum.

Eventually, the convergent pattern will form some square $s_i$ where the entire region inside the square has a constant weighting. At this point, the weighting of the region can be disregarded and the discretization of the graph can be stopped. We propose a simple modification to Weiszfeld's algorithm that takes into account the different weightings of each individual node. In this regard, the goal point $q$ from here is known as the geometric median \cite{geometricmedian}, which is defined as
$$q \in\argmin_{\forall points \exists\in s_i}f(q)\ \textnormal{where}\ f(q) = \sum_{n \in \llbracket 1, 8\rrbracket} w(n) \lVert q - c(\theta_n) \rVert_2$$
where $w(n)$ represents the weighting of some node $\theta_n$, $c(\theta_n)$ represents the coordinate position of the related successor node to $w(n)$ ($\in \mathbb{R}^2$)  and $s_i$ is the convex hull of points $c(\theta_n)$ for $n \in \llbracket 1, 8 \rrbracket$.

\begin{lemma}
Call some initial point inside square $s_i$ $q_1$. Continuously iterating through
$$q_{i+1} = \left( \sum_{n=1}^{8} \frac{w(n) \cdot c(\theta_n)}{\lVert q_i - c(\theta_n) \rVert} \right) \Bigg/ \left( \sum_{n=1}^{8} \frac{w(n)}{\lVert q_i - c(\theta_n) \rVert} \right)$$
will approach the geometric median $q$: 
$$ \lim_{i \to \infty} q_i = q.$$
\end{lemma}

\begin{figure}[h]
\begin{center}
\includegraphics[scale=.28]{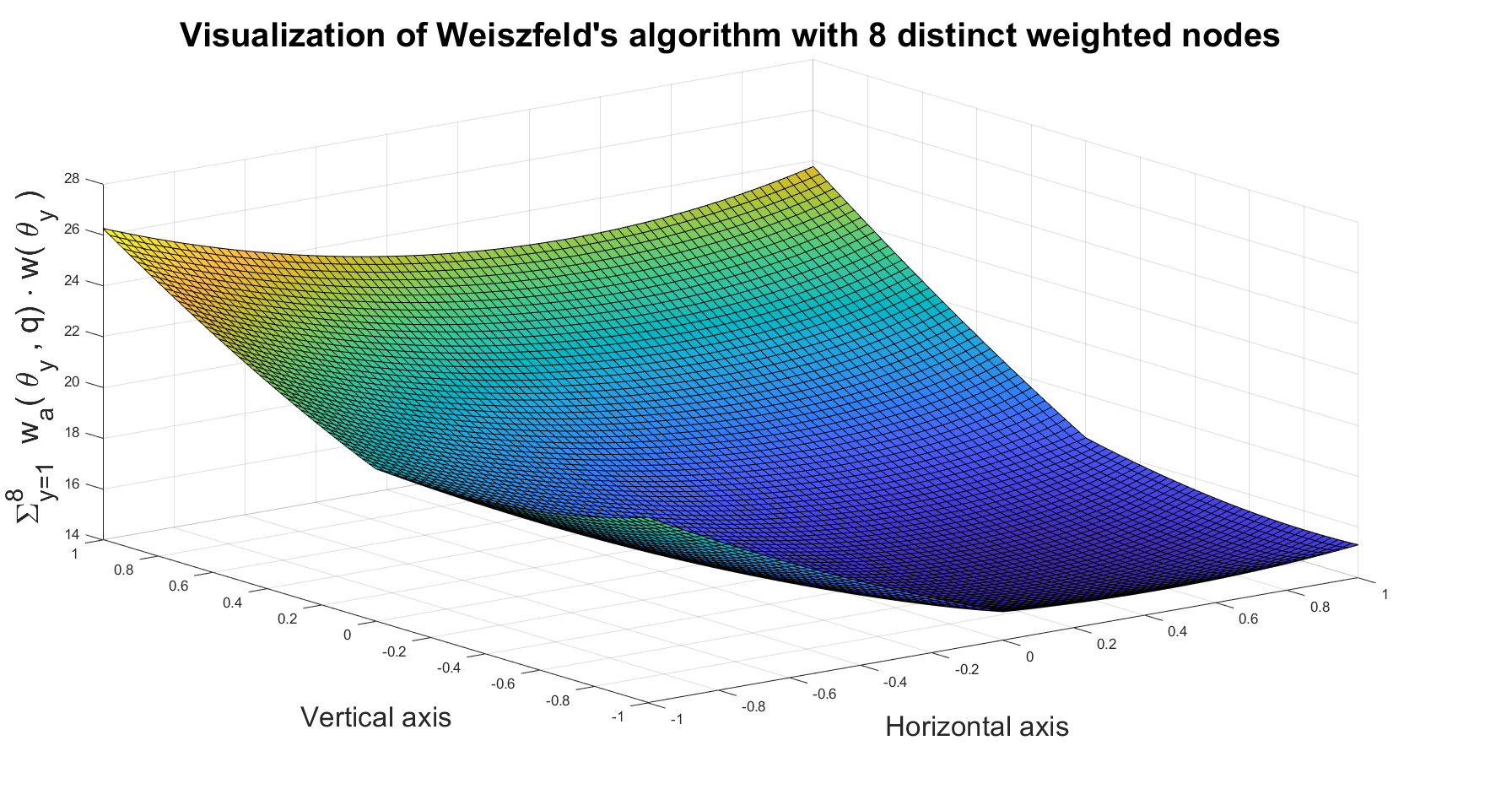}
\caption{Manifold produced by calculating the distance to each $c(\theta_y)$ for $y \in \llbracket 1, 8 \rrbracket$. From $(-1, 1)$ and going clockwise, the weightings are as follows: $1, 0, 1, 2, 4, 2, 1, 2$. Weiszfeld's algorithm would approximate the minimum along this manifold}
\end{center}
\end{figure}

\newpage
\begin{remark}
Note that since the Weighted Regions algorithm converges towards a single point over an infinite number of iterations and Euclidean space is defined as continuous, any randomly selected singular point that the algorithm converges towards in this space cannot fall on top of a line. Therefore, there will always be a case where Wesizfeld's algorithm can be applied.
\end{remark}

\begin{lemma}
The number of iterations needed to calculate the target accuracy (represented by some variable $\epsilon \propto n_1$) can be modeled by
$$\ceil*{\frac{3\log{2}-2\log{\epsilon_c}}{2\log{\left[\sqrt{\ceil*{\frac{l}{n_1}} \cdot \ceil*{\frac{m}{n_1}}}\right]}-2\log{2}}}+1$$
where $\epsilon_c = \frac{\epsilon}{n_1}$, $\epsilon$ is the target accuracy, and $\left[\sqrt{\ceil*{\frac{l}{n_1}} \cdot \ceil*{\frac{m}{n_1}}}\right]$ represents the number of squares after iteration $y+1$ in $s_y\ \forall y$ (this technique does not take into account the possible utilization of Weiszfeld's algorithm). 
\end{lemma}

\section{Star Topology Obstacles Algorithm}
This section details an algorithm that can compute a minimum for some set of nodes $p_1$, $p_2$,$...$, $p_N$. The technique used is necessary to allow for the calculation of a minimum while running in reasonable time (i.e. not just sampling every possible point across the convex hull). By continuously "zooming" in on a certain sub-section of the convex hull, it becomes possible to focus towards a minimum.

Although the continuous Dijkstra paradigm (CDP) is known to run in polynomial time $O(n \log{n})$ (where $n$ represents the number of distinct vertices along $\mathcal{P}$ and the convex hull) when finding the minimum distance between two points, little is known about the circumstances surrounding minimum distance from a central node to several other nodes with interiors of polygonal obstacles modeled by $\mathcal{P} = \{P_1, P_2, P_3,..., P_k\}$ where $\forall P_i \in \mathcal{P}, P_i \in \mathbb{R}^2$ \cite{cdp2,latexcompanion}. 

Begin by assuming that the paths from some central node $q$ to additional nodes $p_1$, $p_2$,$...$, $p_N$ cannot intersect with the interiors of these polygons (and each node is on the Euclidean plane). Begin by creating a convex hull $C \in \mathbb{R}^2$ entrapping all polygons $\mathcal{P}$ and nodes $p_1$, $p_2$,$...$, $p_N$ and placing a grid over the convex hull with each horizontal and vertical line separated by distance $n$. Extremely high values of $n$ result in inaccuracies that could lead towards a shallow local minimum, while extremely low values of $n$ are more likely to zone in on deeper local minimums but with the trade-off of high runtimes. Reform all squares that are only partially in the convex hull into complete squares with side length $n$, and any squares formed completely outside of the convex hull are discarded. From the remaining squares $r$, position an independent node at the center of each one. If any of these nodes fall inside $\mathcal{P}$, disregard the square (however still count it as a square $\in r$). From each one of these nodes $h_i^{(0)}$, simulate the paradigm by propagating an initially circular wavefront that changes on interactions with obstacles. The wavefront from $h_i^{(0)}$ at time $t$ is the set of points $W(d) := \{x \in C|d(h_i^{(0)}, x) = t\}$ where $x$ represents the points that are considered the wavefront at time $t$. Interactions with obstacles result in the wavefront being formed of multiple wavelets, each of which have a center from $h_i$ or from a vertex along $\mathcal{P}, C$. While the waves propagate, the distinct wavelets may be generated, disappeared, or broken into two wavelets \cite{latexcompanion}. When the two neighboring wavelets to some arbitrary wavelet merge together, the middle wavelet can be said to be eliminated. Any existing wavelet upon contact with $\mathcal{P}$ for the first time splits into two new wavelets. 

The first intersection between $W(d)$ and $p_j$ for each $j \in \llbracket 1, N \rrbracket$ has a related $d(h_i^{(0)}, p_j)$. For each $h_i^{(0)}$, 
$$\sum_{j=1}^{N}d(h_i^{(0)}, p_j)$$
For the minimum of all these values, create a square $s_1$ around it connecting the 8 surrounding squares. If an edge or corner square is selected for the center of $s_1$, create "empty" squares around the square chosen such that there exists 8 surrounding nodes. Although the grid is placed over these "empty" squares, any squares created after the grid is placed that are fully in the "empty" squares are disregarded (i.e. their points $h_i^{(0)}$ are not tested). Parse $s_1$ into $\left[\rule{0cm}{.3cm} \sqrt{r} \right]\rule{0cm}{.3cm}^2$ new squares. Create new nodes at the center of each new square called $h_i^{(1)}$. From each one of these nodes, repeat through the pattern outlined earlier (however, if $h_i^{(y)} \in \mathcal{P}$, don't test the point instead of removing it). With every new, smaller square ($s_y$) being formed, it has nodes $h_i^{(y)}$ for $i \in \left\llbracket 1, \left[\rule{0cm}{.3cm} \sqrt{r} \right]\rule{0cm}{.3cm}^2\right\rrbracket$. Continue until some level of accuracy $\epsilon$ where $\epsilon \propto n$ and $\epsilon \in (0, 2n\sqrt{2}]$. 

\begin{lemma}
The number of iterations needed to calculate the accuracy of $\epsilon$ can be determined by
$$\ceil*{\frac{3\log{2}-2\log{\epsilon_c}}{2\log{\left[\rule{0cm}{.3cm} \sqrt{r} \right]\rule{0cm}{.3cm}}-2\log{2}}}+1$$
where $\epsilon_c = \frac{\epsilon}{n}$, $\epsilon$ is the target accuracy and $\left[\rule{0cm}{.3cm} \sqrt{r} \right]\rule{0cm}{.3cm}^2$ represents the number of squares after iteration $y+1$ in $s_y\ \forall y$.
\end{lemma}

\newpage

\begin{figure}[h]
\begin{center}
\includegraphics[scale=.2781]{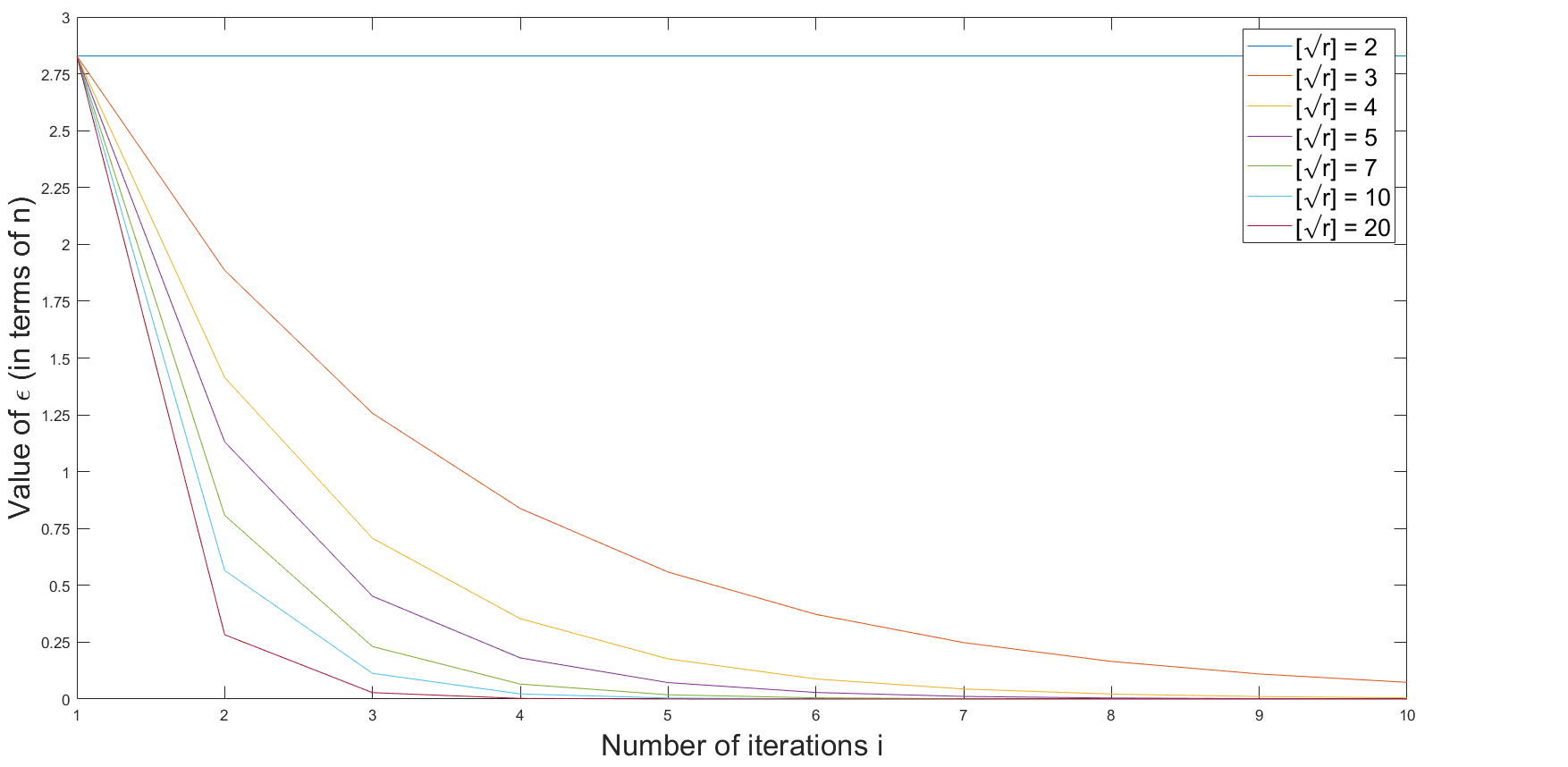}
\caption{Rate of convergence for distinct values of $\left[\rule{0cm}{.3cm} \sqrt{r} \right]\rule{0cm}{.3cm}$. If $\left[\rule{0cm}{.3cm} \sqrt{r} \right]\rule{0cm}{.3cm} \leq 2$, iterations of the algorithm will not converge}
\end{center}
\end{figure}

\begin{remark}
To shorten the notation for runtime complexity in this algorithm, we set
$\left[\rule{0cm}{.3cm} \sqrt{r} \right]\rule{0cm}{.3cm} = 4$. This was based on figure 8, where $\left[\rule{0cm}{.3cm} \sqrt{r} \right]\rule{0cm}{.3cm} = 4$ was determined to be a fair rate of convergence that would maintain a high level of accuracy. We believe it is fair to make these substitutions for the obstacles algorithm and not the weighted regions algorithm because there is greater variation with respects to accuracy in weighted regions, while this variation is more negligible with the obstacles case. 
\end{remark}

\begin{lemma}
The Star Topology Obstacles algorithm given some $\epsilon$ target accuracy in $\mathbb{R}^2$ runs in $O(n\log{n}\log{\frac{1}{\epsilon_c}})$ time
\end{lemma}

\newpage

\section{Overview of both Algorithms}

\begin{algorithm}[ht]
\caption{Star Topology Weigheted Regions Algorithm}\label{AAAAAAAAAAA}
\hspace*{\algorithmicindent} \textbf{Input}: points $q$, $p_1$, $p_2$,$...$, $p_N$ $\in \mathbb{R}^2$\\
\hspace*{\algorithmicindent} \textbf{Input}: set of polygonal regions $\mathcal{P}=\{P_1, P_2, P_3,..., P_k\} \in \mathbb{R}^2$\\
\hspace*{\algorithmicindent} \textbf{Input}: desired accuracy $\epsilon \in (0, 2n_1\sqrt{2}]$
\begin{algorithmic}[1]
\State Create convex hull $C$ such that $\forall p_i \in \mathbb{R}^2, p_i \in C$ 
\State Place a rectangle with dimensions $l$, $m$ such that $C$ is inscribed
\State Position a two-dimensional grid equally spaced some $n_1$ distance apart (and is able to extend over edges $l$, $m$) where all squares have a node $h_i^{(0)}$ at its center
\State $y, k = 0$ 
\While{$k\ != i = \ceil*{\frac{3\log{2}-2\log{\epsilon_c}}{2\log{\left[\sqrt{\ceil*{\frac{l}{n_1}} \cdot \ceil*{\frac{m}{n_1}}}\right]}-2\log{2}}}+1$}
    \For{$i \in \left\llbracket 1, {\left[\rule{0cm}{.225cm}
\sqrt{\ceil*{\frac{l}{n_i}} \cdot \ceil*{\frac{m}{n_i}}}
\right]\rule{0cm}{.225cm}}^2\right\rrbracket$} 
     \State calculate $\sum_{j=1}^{N}w_a(h_i^{(y)}, p_j)$
    \EndFor 
     \State $\forall h_i^{(y)}$, find $\min \sum_{j=1}^{N}w_a(h_i^{(y)}, p_j)$
     \State create square $s_{y+1}$ around $\min h_i^{(y)}$ with the 8 surrounding $h_i^{(y)}$ nodes
     \State cut $s_{y+1}$ into ${\left[\rule{0cm}{.225cm} \sqrt{\ceil*{\frac{l}{n_i}} \cdot \ceil*{\frac{m}{n_i}}}
\right]\rule{0cm}{.225cm}}^2$ new squares
     \State place $h_i^{(y+1)}$ at the center of each new square
\State $y\gets y + 1$
\State $k\gets k + 1$    
\EndWhile
\State \textbf{Output} $\min \sum_{j=1}^{N}w_a(h_i^{(y)}, p_j)$
\end{algorithmic}
\end{algorithm}

\begin{remark}
Note that on line 6 (in algorithm 1) for the first iteration, 
$$i \in \left\llbracket \rule{0cm}{.5cm} 1, 
\ceil*{\frac{l}{n_i}} \cdot \ceil*{\frac{m}{n_i}} \right\rrbracket \rule{0cm}{.5cm}.$$
\end{remark}

\begin{remark}
For the sake of brevity we choose not to include the possible Weiszfeld's algorithm addition, which can serve as an alternative to the while loop once the region inside some square $s_i$ contains the same weighting. 
\end{remark}

\newpage
\begin{algorithm}[ht]
\caption{Star Topology Obstacles Algorithm}\label{AAAAAAAAAAA}
\hspace*{\algorithmicindent} \textbf{Input}: points $q$, $p_1$, $p_2$,$...$, $p_N$ $\in \mathbb{R}^2$\\
\hspace*{\algorithmicindent} \textbf{Input}: set of polygonal obstacles $\mathcal{P} = \{P_1, P_2, P_3,..., P_k\} \in \mathbb{R}^2$\\
\hspace*{\algorithmicindent} \textbf{Input}: desired accuracy $\epsilon \in (0, 2n\sqrt{2}]$
\begin{algorithmic}[1]
\State Create $C$ such that $\forall p_i, P_j \in \mathbb{R}^2, p_i, P_j \in C$
\State Position a two-dimensional grid equally spaced some $n$ distance apart where all squares who have some part $\in C$ and some part $\not\in \mathcal{P}$ have a node $h_i^{(0)}$ at its center
\State $y, k = 0$ 
\While{$k\ != i = \ceil*{\frac{3\log{2}-2\log{\epsilon_c}}{2\log{\left[\rule{0cm}{.3cm} \sqrt{r} \right]\rule{0cm}{.3cm}} - 2\log{2}}}+1$}
    \For{$i \in \left\llbracket 1, \left[\rule{0cm}{.3cm} \sqrt{r} \right]\rule{0cm}{.3cm}^2\right\rrbracket$}
     \State calculate $\sum_{j=1}^{N}d(h_i^{(y)}, p_j)$
    \EndFor 
     \State $\forall h_i^{(y)}$, find $\min \sum_{j=1}^{N}d(h_i^{(y)}, p_j)$
     \State create square $s_{y+1}$ around $\min h_i^{(y)}$ with the 8 surrounding $h_i^{(y)}$ nodes
     \State cut $s_{y+1}$ into $\left[\rule{0cm}{.3cm} \sqrt{r} \right]\rule{0cm}{.3cm}^2$ new squares
     \State place $h_i^{(y+1)}$ at the center of each new square
\State $y\gets y + 1$
\State $k\gets k + 1$    
\EndWhile
\State \textbf{Output} $\min \sum_{j=1}^{N}d(h_i^{(y)}, p_j)$
\end{algorithmic}
\end{algorithm}

\section{Conclusion}

We initialize this paper by proposing boundings on the possible positions of the central node for minimum distance star topology graphs across three distinct cases. We additionally develop two novel algorithms, dealing with specific cases of star topology when there is the inclusion of obstacles and weighted regions, both of which are achieved by discretizing the Euclidean plane. Overall we believe these algorithms will provide new opportunities in civil engineering, finding cheaper routes to connect central power stations to multiple external cities. Although these algorithms are still in their early stages, we are interested in finding ways to  incorporate them into programming languages to develop manifolds similar to the ones obtained for Weiszfeld's algorithm, while also looking into possible runtimes for certain cases of the first algorithm (although this may not necessarily be calculable as A* pathfinding runs on a heuristic).

\newpage

\medskip

\bibliographystyle{plainurl}
\bibliography{name}

\appendix 
\section{Properties of Convex Hulls (Proofs)}
This section compiles all proofs for lemmas stated in section 4 on convex hulls surrounding Euclidean spaces, spaces with obstacles, and spaces parsed into weighted regions. 

\subsection{Bounding for Algorithms in Euclidean Space}

\begin{lemma}
For point $q \in \mathbb{R}^2$ that falls on some side of line $o$ and points $p_1, p_2,..., p_N \in \mathbb{R}^2$ that fall on the same side of line $o$ as $q$ or on line $o$, 
$$d(p_{i}, q') \geq d(p_{i}, q)$$
where $q^\prime$ is the reflection of $q$ across $o$.
\end{lemma}

\begin{proof}
Construct line segments $\overline{qp_i}$, $\overline{q^{\prime}p_i}$, $\overline{qq^\prime}$. Since $q^\prime$ is a reflection over $o$, it logically follows that $\overline{qq^\prime} \bot o$. Place a new point $p_i^\prime$ on $\overline{qq^\prime}$ such that $\overline{p_ip_i^\prime} \bot \overline{qq^\prime}$ and thus,  $\overline{p_ip_i^\prime} \parallel \overline{o}$. Since $p_i^\prime$ falls either on the same side as $q$ or on $o$, 

$$\overline{p_i^{\prime}q^{\prime}} \geq \overline{p_i^{\prime}q}$$
Since all line segments are strictly a positive value, 
$$\overline{p_i^{\prime}q^{\prime}}^2 \geq \overline{p_i^{\prime}q}^2$$
$$\overline{p_i^{\prime}q^{\prime}}^2 + \overline{p_ip_i^\prime}^2 \geq \overline{p_i^{\prime}q}^2 + \overline{p_ip_i^\prime}^2$$
$$\overline{p_iq^{\prime}}^2 \geq \overline{p_iq}^2$$
by Pythagorean theorem. On the Euclidean plane without any weighted regions, obstacles, or other obtrusions, the minimum distance between two points is a straight line, and therefore 
$$d(p_{i}, q') \geq d(p_{i}, q)$$
must be true, concluding the proof. A visual example is provided by figures 1 and 9, where figure 1 represents the first reflection and figure 9 represents the second reflection. 
\end{proof}

\begin{figure}[h]
\begin{center}
\includegraphics[scale=.3]{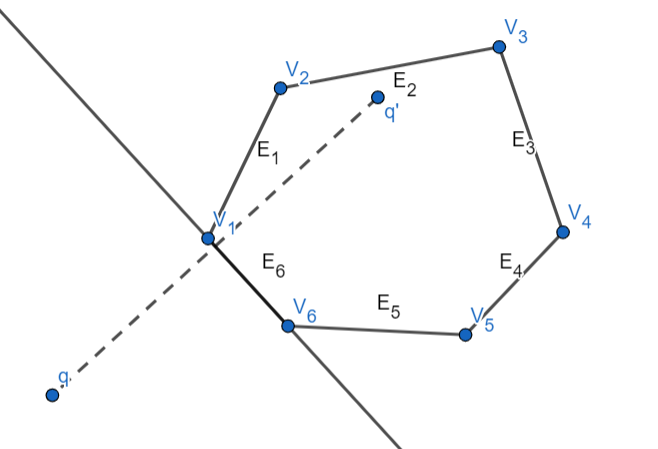}
\caption{Point $q := q^\prime$ mirrored over line $o$ which extends out from line segment $E_6$. New $q^\prime$ falls within the convex hull}
\end{center}
\end{figure}

\begin{theorem}
For some set of $N+1$ points $q$, $p_1$, $p_2$,$...$, $p_N \in \mathbb{R}^2$, 
$$\min \sum_{i=1}^{N}d(q, p_i)$$ 
occurs when $q$ falls in or along $C$, where $C$ is the convex hull formed by points $p_1$, $p_2$,$...$, $p_N$.
\end{theorem}

\begin{proof}
First assume point $q$ falls outside of the convex hull $C$. Select the two vertices on the exterior of the convex hull that are closest to $q$. Extend the line segment formed by these two vertices into line $o$. Call the reflection of $q$ across this line point $q^\prime$. By lemma 4.6, it is known that $d(p_{i}, q') \geq d(p_{i}, q)$. It follows that 
$$\sum_{i=1}^{N}d(p_{i}, q') \geq \sum_{i=1}^{N}d(p_{i}, q)$$
However, since the points of a convex hull can never be strictly co-linear, the points $p_1$, $p_2$,$...$, $p_N$ will never all fall on $o$, so there will always be some circumstance where $d(p_{i}, q') > d(p_{i}, q)$, therefore 
$$\sum_{i=1}^{N}d(p_{i}, q') > \sum_{i=1}^{N}d(p_{i}, q)$$
For each new point $q^{\prime}$, assume $q := q^\prime$. Then reiterate through this strategy to a new point $q^\prime$. This pattern repeats acyclically, as once $q^\prime$ mirrors over some edge $o$ and ends up inside the convex hull, 
$$\sum_{i=1}^{N}d(p_{i}, q') < \sum_{i=1}^{N}d(p_{i}, q)$$
is instead true as a result of $q^\prime$ being on the same side of $o$ as points $p_1$, $p_2$,$...$, $p_N$. After performing $q := q^\prime$, the new position of $q^\prime$ falls outside of the convex hull. This pattern repeats cyclically, and thus the point $q^\prime$ (or $q$) inside the convex hull serves as 
$$\min \sum_{i=1}^{N}d(q, p_i)$$ 
for all $q$ or $q^\prime$. Therefore, any points outside the convex hull has a point inside the convex hull that is closer to the set of points $p_1$, $p_2$,$...$, $p_N$, and thus, a contradiction arises, concluding the proof. 
\end{proof}

\subsection{Bounding for Algorithms in Region with Obstacles}

\begin{lemma}
For some set of points $p_1$, $p_2$,$...$, $p_N$ and set of polygons whose interior can be modeled by $\mathcal{P} = \{P_1, P_2, P_3,..., P_k\}$ that all fall on or on one side of some line $o$ in the Euclidean plane, the minimum sum of distances from a new point $q$ to the set of points $p_1$, $p_2$,$...$, $p_N$ such that 
$\forall P_j \in \mathcal{P}, P_j \sqcap p_i = \emptyset$ and $P_j \sqcap \ell(d(q, p_i)) = \emptyset$ for $i \in \llbracket 1, N \rrbracket$ must occur where $q$ is either on line $o$ or on the same side as the aforementioned set of points.
\end{lemma}

\begin{proof}
First assume a point $q_1$ falls on the opposite side of $o$ as points $p_1$, $p_2$,$...$, $p_N$ such that $$\sum_{i=1}^{N}d(p_i, q_1)$$
is minimized. Call the intersection of $\ell(d(p_i, q_1))$ and $o$ point $C_i$ for $i \in \llbracket 1, N \rrbracket$. Then, position a point $q_2$ on line $o$ such that $\overline{q_1q_2}\ \bot \ o$. It follows that 
$$\sum_{i=1}^{N}d(C_{i}, q_{2}) < \sum_{i=1}^{N}d(C_{i}, q_1)$$
as $d(C_{i}, q_{2})$ is a leg of a right triangle while $d(C_{i}, q_1)$ is the hypotenuse. Since $q_1 \sqcap o = \emptyset$, the other leg of this right triangle must traverse some distance, therefore demonstrating that $d(C_i, q_2) < d(C_i, q_1)$ for $i \in \llbracket 1, N \rrbracket$. 
Since $d(p_i, q_2) \leq d(p_i, C_i, q_2)$ (as $C_i$ either falls along the path or serves as a detour), $d(p_i, q_2) \leq d(p_i, C_i, q_2) = d(p_i, C_i) + d(C_i, q_2) < d(p_i, C_i) + d(C_i, q_1) = d(p_i, q_1)$ 
for $i \in \llbracket 1, N \rrbracket$. Therefore

$$\sum_{i=1}^{N}d(p_i, q_{2}) < \sum_{i=1}^{N}d(p_i, q_1)$$
and thus a contradiction arises, concluding the proof. 
\end{proof}

\begin{theorem}
For a set of polygons whose interiors can be modeled by $\mathcal{P} = \{P_1, P_2, P_3,..., P_k\}$ and $N+1$ points $q$, $p_1$, $p_2$,$...$, $p_N$ where 
$\forall P_j \in \mathcal{P}, P_j \sqcap p_i = \emptyset$ and $P_j \sqcap \ell(d(C_i, p_i)) \neq \emptyset$  for $i \in \llbracket 1, N \rrbracket$, 
$$\min \sum_{i=1}^{N}d(p_{i}, q)$$ 
must occur when $q$ falls in or along $C$, where $C$ is the convex hull entrapping all points $p_1$, $p_2$,$...$, $p_N$, and the set of polygons $\mathcal{P}$.
\end{theorem}

\begin{proof}
Assume the convex hull $C = (V, E)$ has edges and vertices labeled clockwise cyclically from $V_1$, $V_2$,$...$, $V_N$, $V_{N+1}$ where $V_{N+1}=V_1$ and edges $E_i$ between $V_i$ and $V_{i+1}$. For each $V_i$ where $i \in \llbracket 1, N \rrbracket$, draw lines outward that are perpendicular to $E_i$, $E_{i+1}$ at $V_i$. Assume some point $q_1$ does not fall within or on an edge of $C$. Then, $q_1$ must fall within one of four spaces (shown in figure 10): \\

\begin{figure}[h]
\begin{center}
\includegraphics[scale=.46]{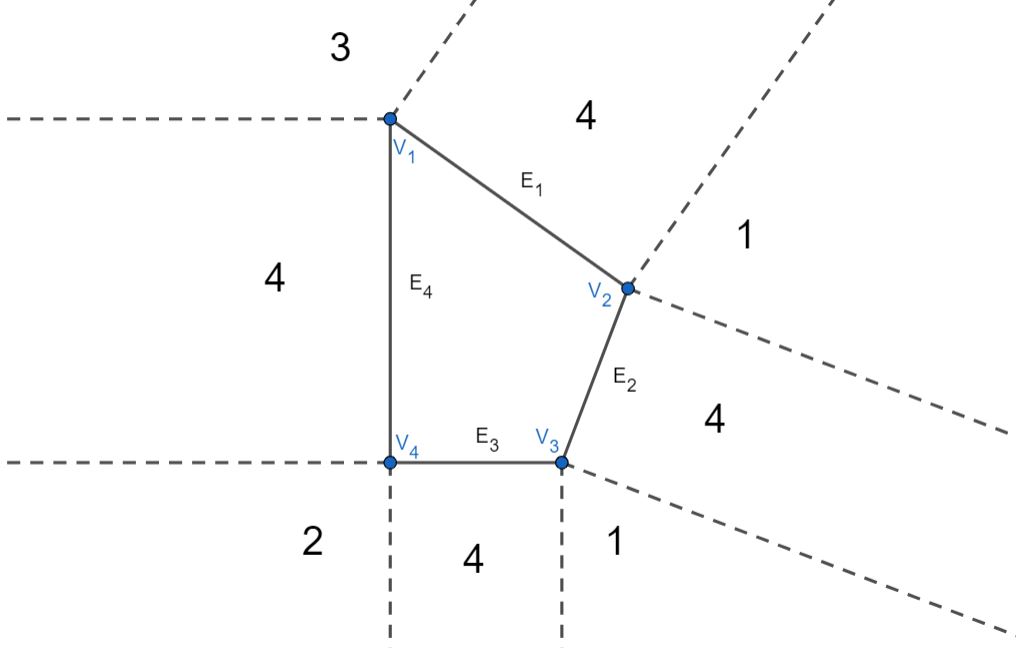}
\caption{Parsing of area outside convex hull into 4 cases}
\end{center}
\end{figure}

Case 1: $q_1$ falls within an area formed by an acute angle \\
Assume the point and edges nearest to $q_1$ are $V_i$ and $E_i$, $E_{i+1}$. Extend edges $E_i$, $E_{i+1}$ outward. Without loss of generality position a point $q_2$ along the extension of $E_{i+1}$ where $\overline{q_1q_2}\ \bot \ E_{i+1}$. Since all points fall inside or along the convex hull (and thus on the other side of $E_{i+1}$), lemma 4.10 can then be applied. This process can be iterated again to a point $q_3$ that falls along $E_i$ (since $\overline{q_1q_2}\ \bot E_{i}$ and all points fall on the other side of $E_i$, lemma 4.10 can again be applied). It follows that 
$$\sum_{i=1}^{N}d(p_i, q_{3}) < \sum_{i=1}^{N}d(p_i, q_2) < \sum_{i=1}^{N}d(p_i, q_1)$$
with $q_3$ being along the convex hull.\\
Case 2: $q_1$ falls within an area formed by a right angle \\
A similar strategy to Case 1 can be applied, with $V_i$ and $E_i$, $E_{i+1}$ being the closest vertex and edges to $q_1$, respectively. Without loss of generality, assume a point $q_2$ falls along $E_i$ such that $\overline{q_1q_2}\ \bot \ E_{i}$. Then move along the extension of $E_i$ to a point $q_3$ that is located at $V_i$. This holds true since $E_i$ is perpendicular to $E_{i+1}$. In both scenarios, lemma 4.10 proves that both these movements are strictly decreasing 
$$\sum_{i=1}^{3}d(p_i, q_{i})$$ \\
Case 3: $q_1$ falls within an area formed by an obtuse angle \\
Call the nearest vertex to $q_i$ $V_i$ and the nearest edges $E_i$, $E_{i+1}$. Extend line segments $E_i$, $E_{i+1}$ outward to $o_i$, $o_{i+1}$ respectively. If $q_1$ now falls between these two extensions, without loss of generality, call a point $q_2$ on $o_i$ such that $\overline{q_1q_2}\ \bot \ o_{i}$. If $q_1$ does not fall between $o_i$, $o_{i+1}$, determine whether is closer to $o_i$ or $o_{i+1}$. If $q_1$ falls closer to $o_{i+1}$, place point $q_2$ on $o_i$ such that $\overline{q_1q_2}\ \bot \ o_{i}$. Alternatively, if $q_1$ falls closer to $o_{i}$, create a line through $q_1$ orthogonal to $o_{i+1}$. Call the intersection of this line and $o_i$ point $q_2$. Call a point $q_3$ on $o_{i+1}$ such that $\overline{q_2q_3}\ \bot \ o_{i+1}$. Continuously reiterate through this process to a new point $q_j$ that follows the piece-wise function
\[
  F(i) =
  \begin{cases} 
                                   \overline{q_{j-1}q_{j}}\ \bot \ o_{i} & \text{if $ j \equiv 0\ (mod\ 2)$} \\
                                   \overline{q_{j-1}q_{j}}\ \bot \ o_{i+1} & \text{if $ j \equiv 1\ (mod\ 2)$} \\
  \end{cases}
\]
These iterations and steps proposed earlier all follow lemma 4.10 as the set of points $p_1$, $p_2$,$...$, $p_N$ all fall on the opposite side or along the line that is being moved towards perpendicularly. Therefore,
$$\lim_{j\to\infty} q_j = V_i$$
Case 4: $q_1$ falls within an area enclosed by 2 right angles or along a boundary between two areas \\
Call the edge nearest to $q_1$ $E_i$. Position $q_2$ along $E_i$ (and thus on  the convex hull) such that $\overline{q_1q_2}\ \bot \ E_{i}$. lemma 4.10 proves that 
$$\sum_{i=1}^{N}d(p_i, q_{2}) < \sum_{i=1}^{N}d(p_i, q_{1})$$ \\

By proving that no point $q_1$ outside the convex hull could minimize the distance to points $p_1$, $p_2$,$...$, $p_N$ as compared to a point either on or inside the convex hull, a contradiction arises, concluding the proof. 
\end{proof}

\subsection{Bounding for Algorithms in Weighted Regions}

\begin{lemma}
For some finite region $\in \mathbb{R}^2$ parsed into $k$ distinct polygons $\mathcal{P} = \{P_1, P_2, P_3,..., P_k\}$ each with an independent weighting $w(k)$ and $N+1$ points $q$, $p_1$, $p_2$,$...$, $p_N$, such that $p_1$, $p_2$,$...$, $p_N$ and $\mathcal{P}$ fall on the same side or on some line $o$ in the Euclidean plane,
$$\min \sum_{i=1}^{N}w(p_i, q)$$
occurs when $q$ falls in or on $C$ where C is the space of the convex hull entrapping all points $p_1$, $p_2$,$...$, $p_N$, and the polygons $\mathcal{P}$.
\end{lemma}

\begin{proof}
lemma 4.10 can be modified, where $w(q, p_i)$ is similar conceptually to $d(q, p_i)$ for $i \in \llbracket 1, N \rrbracket$. Likewise, call $C_i$ the intersection of $\ell(w(q, p_i))$ and $o$. Furthermore, position a point $q_2$ on $o$ such that $\overline{q_1q_2}\ \bot \ o$. Triangles can then be constructed between points $q_1$, $q_2$, and $C_i$ for $i \in \llbracket 1, N \rrbracket$. Each right triangle $\bigtriangleup C_iq_1q_2$ has hypotenuse $\overline{q_1C_i}$ and legs $\overline{q_2q_1}$, $\overline{q_2C_i}$. Based on the properties of right triangles, it is known that $\overline{q_1C_i} > \overline{q_2C_i}\ \forall C_i$. For all lengths that fall in non-weighted space, $\ell(d(x, y)) = \overline{xy}$. Since all sides of triangle $\bigtriangleup C_iq_1q_2$ fall in non-weighted space, 
$$\sum_{i=1}^{N}d(q_1, C_i) > \sum_{i=1}^{N}d(q_2, C_i)$$
It logically follows that
$$\sum_{i=1}^{N}d(q_1, C_i) + w(C_i, p_i) > \sum_{i=1}^{N}d(q_2, C_i) + w(C_i, p_i)$$
$$\sum_{i=1}^{N}w(q_1, p_i) > \sum_{i=1}^{N}d(q_2, C_i) + w(C_i, p_i)$$
where $\ell(w(C_i, p_i)), \ell(d(q_1, C_i))$ are both parts of the path $\ell(w(q_1, p_i))$. 
Furthermore, 
$$\sum_{i=1}^{N}d(q_2, C_i) + w(C_i, p_i) \geq \sum_{i=1}^{N} w(q_2, p_i)$$
as $\ell((q_2, C_i)) + \ell(w(C_i, p_i))$ always traverses through point $C_i$, serving as a possible detour from $\ell(w(q_2, p_i))$.
Therefore, 
$$\sum_{i=1}^{N}w(q_1, p_i) > \sum_{i=1}^{N} w(q_2, p_i)$$
and thus a contradiction arises, proving that no point $q$ on the opposite side of line $o$ as points $p_1$, $p_2$,$...$, $p_N$ and $\mathcal{P}$ in the Euclidean plane can fulfill
$$\min \sum_{i=1}^{N}w(p_i, q).$$
\end{proof}

\begin{theorem}
For some finite region $\in \mathbb{R}^2$ parsed into $k$ polygons $\mathcal{P} = \{P_1, P_2, P_3,..., P_k\}$ each with an independent weighting $w(k)$ and $N+1$ points $q$, $p_1$, $p_2$,$...$, $p_N$, 
$$\min \sum_{i=1}^{N}w(p_i, q)$$
occurs when $q$ falls in or along $C$, where C is the space of the convex hull entrapping all points $p_1$, $p_2$,$...$, $p_N$, and the set of polygons $\mathcal{P}.$
\end{theorem}

\begin{proof}
Similarly to theorem 4.11, assume the convex hull $C = (V, E)$ has edges and vertices labeled clockwise cyclically from $V_1$, $V_2$,$...$, $V_N$, $V_{N+1}=V_1$ and edges $E_i$ between $V_i$ and $V_{i+1}$. For each $V_i$ where $i \in \llbracket 1, N \rrbracket$, draw lines outward such that they are perpendicular to $E_i$, $E_{i+1}$ at $V_i$. Assume some point $q_1 \not\in C$. Then, $q_1$ must fall within one of four spaces as stated in theorem 4.11. All procedures in all 4 cases can likewise be repeated, however instead of applying lemma 4.10, lemma 4.13 can instead be substituted in, and by following each case described in theorem 4.11, a similar result can thus be achieved.
\end{proof}

\section{Star Topology Weighted Regions}

\begin{lemma}
Introduction of Weiszfeld's algorithm as an alternative to the algorithm proposed in section 5. \\
Call some initial point inside square $s_i$ $q_1$. Continuously iterating through
$$q_{i+1} = \left( \sum_{n=1}^{8} \frac{w(n) \cdot c(\theta_n)}{\lVert q_i - c(\theta_n) \rVert} \right) \Bigg/ \left( \sum_{n=1}^{8} \frac{w(n)}{\lVert q_i - c(\theta_n) \rVert} \right)$$
will approach the geometric median $q$: 
$$ \lim_{i \to \infty} q_i = q.$$
\end{lemma}

\begin{proof}
Begin by stating a generalized Weiszfeld's algorithm \cite{Eckhardt_1980}
$$q_{i+1} = \left( \sum_{n=1}^{N} \frac{c(\theta_n)}{\lVert q_i - c(\theta_n) \rVert} \right) \Bigg/ \left( \sum_{n=1}^{N} \frac{1}{\lVert q_i - c(\theta_n) \rVert} \right)$$
Which is known to converge towards the geometric median so long as $q_i \not= c(\theta_n)\  \forall i$ and for $n \in \llbracket 1, 8 \rrbracket$ \cite{Love2000}. Because the initial position of $q_1$ must fall inside $s_i$, $q_i \not= c(\theta_n)$ (although it may converge towards $c(\theta_n)$, it will never hit it even if $c(\theta_n)$ is the position of $q$) \cite{Love2000}.
The weighting of a specific point (which is calculated by the number of 
$\ell(w_a(p_i, q))$ that pass through $c(\theta_n)$ for $n \in \llbracket 1, 8 \rrbracket$) is similar to assuming multiple points fall at exactly the same coordinates. $w(n)$ therefore represents the number of $\ell(w_a(p_i, q))$ that pass through $c(\theta_n)$ for each $n \in \llbracket 1, 8 \rrbracket$. Therefore, Weiszfeld's algorithm thus becomes 
$$q_{i+1} = \left( \sum_{n=1}^{N} \sum_{j=1}^{w(n)} \frac{c(\theta_n)}{\lVert q_i - c(\theta_n) \rVert} \right) \Bigg/ \left( \sum_{n=1}^{N} \sum_{j=1}^{w(n)} \frac{1}{\lVert q_i - c(\theta_n) \rVert} \right)$$
which can also be expressed as 
$$q_{i+1} = \left( \sum_{n=1}^{8} \frac{w(n) \cdot c(\theta_n)}{\lVert q_i - c(\theta_n) \rVert} \right) \Bigg/ \left( \sum_{n=1}^{8} \frac{w(n)}{\lVert q_i - c(\theta_n) \rVert} \right).$$
\end{proof}

\begin{lemma}
The number of iterations needed to calculate the target accuracy (represented by some variable $\epsilon \propto n_1$) can be modeled by
$$\ceil*{\frac{3\log{2}-2\log{\epsilon_c}}{2\log{\left[\sqrt{\ceil*{\frac{l}{n_1}} \cdot \ceil*{\frac{m}{n_1}}}\right]}-2\log{2}}}+1$$
where $\epsilon_c = \frac{\epsilon}{n_1}$, $\epsilon$ is the target accuracy, and $\left[\sqrt{\ceil*{\frac{l}{n_1}} \cdot \ceil*{\frac{m}{n_1}}}\right]$ represents the number of squares after iteration $y+1$ in $s_y\ \forall y$ (this technique does not take into account the possible utilization of Weiszfeld's algorithm). 
\end{lemma}

\begin{proof}
After the first iteration square $s_1$ with side length $2n_1$ is created. The maximum inaccuracy of this square is therefore $2n_1\sqrt{2}$. The square $s_2$ with side length $n_2 := \frac{2n_1}{\left[\sqrt{\ceil*{\frac{l}{n_1}} \cdot \ceil*{\frac{m}{n_1}}}\right]}$ is then created, with a maximum possible inaccuracy of $\frac{4n_1\sqrt{2}}{\left[\sqrt{\ceil*{\frac{l}{n_1}} \cdot \ceil*{\frac{m}{n_1}}}\right]}$. As each new square $s_i$ is created, the inaccuracy can be modeled by $\frac{2^in_1\sqrt{2}}{\left[\sqrt{\ceil*{\frac{l}{n_1}} \cdot \ceil*{\frac{m}{n_1}}}\right]^{i-1}}$. This grows exponentially, so to normalize this change, logs are applied. For the proof, We substitute $s_c$ for the square count $\left[\sqrt{\ceil*{\frac{l}{n_1}} \cdot \ceil*{\frac{m}{n_1}}}\right]$. Each $\epsilon$ must fulfill the condition
$$ \frac{2^{i+1}n_1\sqrt{2}}{s_c^{i}} \leq \epsilon < \frac{2^in_1\sqrt{2}}{s_c^{i-1}}$$
and $\epsilon_c = \frac{\epsilon}{n_1}$,
$$ \frac{2^{i+1}\sqrt{2}}{s_c^{i}} \leq \epsilon_c < \frac{2^i\sqrt{2}}{s_c^{i-1}}$$
$$ \log{\frac{2^i \cdot 2\sqrt{2}}{s_c^{i}}} \leq \log{\epsilon_c} < \log{\frac{2^{i-1} \cdot 2\sqrt{2}}{s_c^{i-1}}}$$
$$ \log{2^{3/2} + i(\log{2} - \log{s_c})} \leq \log{\epsilon_c} < \log{2^{3/2} + (i-1)(\log{2} - \log{s_c})}$$
$$ \frac{3}{2}\log{2}-\log{\epsilon_c} \leq i(\log{s_c}-\log{2}) < \frac{3}{2}\log{2}-\log{\epsilon_c} + \log{s_c} - \log{2}$$
$$ 3\log{2}-2\log{\epsilon_c} \leq 2i(\log{s_c}-\log{2}) < 3\log{2}-2\log{\epsilon_c} + 2\log{s_c} - \log{2}$$
$$ \frac{3\log{2}-2\log{\epsilon_c}}{2\log{s_c} - \log{2}} \leq i < \frac{3\log{2}-2\log{\epsilon_c}}{2\log{s_c} - \log{2}} + 1$$
However, for whatever value $i$ is, it must round up to the nearest integer (to make sure that target accuracy is at minimum met), and thus the number of iterations necessary is
$$\ceil*{\frac{3\log{2}-2\log{\epsilon_c}}{2\log{\left[\sqrt{\ceil*{\frac{l}{n_1}} \cdot \ceil*{\frac{m}{n_1}}}\right]}-2\log{2}}}+1$$
where 1 can stay outside the ceiling function since adding/subtracting whole numbers does not change the rounding off of the ceiling function.
\end{proof}

\section{Star Topology Region with Obstacles}

\begin{lemma}
The number of iterations needed to calculate the accuracy of $\epsilon$ can be determined by
$$\ceil*{\frac{3\log{2}-2\log{\epsilon_c}}{2\log{\left[\rule{0cm}{.3cm} \sqrt{r} \right]\rule{0cm}{.3cm}}-2\log{2}}}+1$$
where $\epsilon_c = \frac{\epsilon}{n}$, $\epsilon$ is the target accuracy and $\left[\rule{0cm}{.3cm} \sqrt{r} \right]\rule{0cm}{.3cm}^2$ represents the number of squares after iteration $y+1$ in $s_y\ \forall y$.
\end{lemma}

\begin{proof}
This proof closely mirrors the one outlined in Lemma B.2. After the first iteration the square $s_1$ of side length $2n$ is created. The maximum possible inaccuracy in this square is thus $2n\sqrt{2}$. After the second iteration, the square $s_2$ has side length $\frac{4n}{\left[\rule{0cm}{.3cm} \sqrt{r} \right]\rule{0cm}{.3cm}}$ and maximum possible inaccuracy of $\frac{4n\sqrt{2}}{\left[\rule{0cm}{.3cm} \sqrt{r} \right]\rule{0cm}{.3cm}}$. As each new square $s_i$ is created, the inaccuracy can be modeled by $\frac{2^{i}n\sqrt{2}}{[\sqrt{r}]^{i-1}}$. This grows exponentially, so to normalize this change, logs must be applied. Every $\epsilon$ must fulfill two conditions: $\epsilon_c = \frac{\epsilon}{n}$ and
$$ \frac{2^{i+1}n\sqrt{2}}{\left[\rule{0cm}{.3cm} \sqrt{r} \right]\rule{0cm}{.3cm}^{i}} \leq \epsilon < \frac{2^in\sqrt{2}}{\left[\rule{0cm}{.3cm} \sqrt{r} \right]\rule{0cm}{.3cm}^{i-1}}$$
$$ \frac{2^{i+1}\sqrt{2}}{\left[\rule{0cm}{.3cm} \sqrt{r} \right]\rule{0cm}{.3cm}^{i}} \leq \epsilon_c < \frac{2^i\sqrt{2}}{\left[\rule{0cm}{.3cm} \sqrt{r} \right]\rule{0cm}{.3cm}^{i-1}}$$
$$ \log{\frac{2 \cdot 2^i\sqrt{2}}{\left[\rule{0cm}{.3cm} \sqrt{r} \right]\rule{0cm}{.3cm}^{i}}} \leq \log{\epsilon_c} < \log{\frac{2 \cdot 2^{i-1}\sqrt{2}}{\left[\rule{0cm}{.3cm} \sqrt{r} \right]\rule{0cm}{.3cm}^{i-1}}}$$
$$ \log{2^{3/2} + i\log{2} - i\log{\left[\rule{0cm}{.3cm} \sqrt{r} \right]\rule{0cm}{.3cm}}} \leq \log{\epsilon_c} < \log{2^{3/2} + (i-1)\log{2} - (i-1)\log{\left[\rule{0cm}{.3cm} \sqrt{r} \right]\rule{0cm}{.3cm}}}$$
$$ \frac{3}{2}\log{2}-\log{\epsilon_c} \leq i(\log{\left[\rule{0cm}{.3cm} \sqrt{r} \right]\rule{0cm}{.3cm}} - \log{2})< \frac{3}{2}\log{2}-\log{\epsilon_c} + \log{\left[\rule{0cm}{.3cm} \sqrt{r} \right]\rule{0cm}{.3cm}} - \log{2}$$
$$ 3\log{2}-2\log{\epsilon_c} \leq 2i(\log{\left[\rule{0cm}{.3cm} \sqrt{r} \right]\rule{0cm}{.3cm}} - \log{2}) < 3\log{2}-2\log{\epsilon_c} + 2(\log{\left[\rule{0cm}{.3cm} \sqrt{r} \right]\rule{0cm}{.3cm}} - \log{2})$$
$$ \frac{3\log{2}-2\log{\epsilon_c}}{2\log{\left[\rule{0cm}{.3cm} \sqrt{r} \right]\rule{0cm}{.3cm}} - 2\log{2}}  \leq i < \frac{3\log{2}-2\log{\epsilon_c}}{2\log{\left[\rule{0cm}{.3cm} \sqrt{r} \right]\rule{0cm}{.3cm}} - 2\log{2}} + 1$$
However, iterations can only be integers, thus
$$\ceil*{\frac{3\log{2}-2\log{\epsilon_c}}{2\log{\left[\rule{0cm}{.3cm} \sqrt{r} \right]\rule{0cm}{.3cm}} - 2\log{2}}} \leq i < \ceil*{\frac{3\log{2}-2\log{\epsilon_c}}{2\log{\left[\rule{0cm}{.3cm} \sqrt{r} \right]\rule{0cm}{.3cm}} - 2\log{2}}}+1$$
1 can stay outside the ceiling function since adding/subtracting whole numbers does not affect the rounding off of the ceiling function. Since the accuracy given must always be higher than the target accuracy, 
$$i = \ceil*{\frac{3\log{2}-2\log{\epsilon_c}}{2\log{\left[\rule{0cm}{.3cm} \sqrt{r} \right]\rule{0cm}{.3cm}} - 2\log{2}}}+1$$
\end{proof}

\subsection{Complexity Analysis}

\begin{lemma}
The Star Topology Obstacles algorithm given some $\epsilon$ target accuracy in $\mathbb{R}^2$ runs in $O(n\log{n}\log{\frac{1}{\epsilon_c}})$ time
\end{lemma}

\begin{proof}
After the first iteration, the number of test points being put through the continuous Dijkstra paradigm has an upper bound of $\left[\rule{0cm}{.3cm} \sqrt{r} \right]\rule{0cm}{.3cm}^2= 4^2 = 16$. As the number of iterations increases, the $\frac{\textnormal{number of test points}}{\textnormal{total trials}}$ continues to maintain an upper bounding of 16. Each test point runs in $O(n \log{n})$ time through the continuous Dijkstra paradigm \cite{latexcompanion}, with $i$ iterations needed to reach some level of accuracy $\epsilon$ where
$$i = \ceil*{\frac{3\log{2}-2\log{\epsilon_c}}{2\log{\left[\rule{0cm}{.3cm} \sqrt{r} \right]\rule{0cm}{.3cm}} - 2\log{2}}}+1.$$
For target accuracy values that are extremely precise, the ceiling function and all constants can be disregarded (note that since $\left[\rule{0cm}{.3cm} \sqrt{r} \right]\rule{0cm}{.3cm} = 4$, it is also a constant), leaving the expression $-\log{\epsilon_c}$ which can also be expressed as $\log{\frac{1}{\epsilon_c}}$. Multiplying the number of total iterations necessary to compute the target accuracy, the number of test points per square, and the algorithm running from each test point proves the algorithm has an upper bound runtime of $O(n\log{n}\log{\frac{1}{\epsilon_c}})$ in $\mathbb{R}^2$ (when disregarding the multiplication of the constant 16).
\end{proof}

\end{document}